\documentclass[letterpaper,10pt,conference]{ieeeconf}  

\IEEEoverridecommandlockouts                              
\usepackage{amsmath}
\usepackage{amssymb}
\usepackage{graphicx}
\usepackage{pifont}
\usepackage{braket}
\usepackage{etoolbox}
\usepackage{upgreek}
\usepackage{comment}
\usepackage{stmaryrd}
\usepackage{dsfont}
\usepackage{url}
\let\labelindent\relax
\usepackage[shortlabels]{enumitem}
\usepackage{xcolor}
\usepackage{nccmath}
\usepackage{upgreek}
\usepackage{xfrac}

\usepackage{color}
\usepackage{colordvi}
\usepackage{xspace}

\newtheorem{theorem}{Theorem}[section]
\newtheorem{corollary}[theorem]{Corollary}
\newtheorem{lemma}[theorem]{Lemma}
\newtheorem{proposition}[theorem]{Proposition}
\newtheorem{example}[theorem]{Example}
\newtheorem{remark}[theorem]{Remark}
\newtheorem{conjecture}[theorem]{Conjecture}

\newcommand{\generate}[2]{\langle #1 \rangle_{\mathsf{#2}}}

\newcommand{\mf}[1]{\mathfrak{#1}}

\newcommand{\C}{\mathbb{C}}
\newcommand{\R}{\mathbb{R}}

\newcommand{\iu}{\mathrm{i}\mkern1mu}

\newcommand{\down}{\shortdownarrow}

\newcommand{\diag}{{\rm diag}}

\newcommand{\id}{{\mathds1}}
\newcommand{\conv}{\operatorname{conv}}

\newcommand{\SU}{{\rm SU}}

\newcommand{\tr}{{\rm tr}}

\newcommand{\spec}{{\rm spec}}

\newcommand{\Ad}{{\rm Ad}}
\newcommand{\ad}{{\rm ad}}

\newcommand{\reach}{\mathsf{reach}}
\newcommand{\derv}{\mathsf{derv}}

\newcommand{\ketbra}[1]{\ket{#1}\!\bra{#1}}
\newcommand{\ketbrax}[2]{\ket{#1}\!\bra{#2}}

\renewcommand{\epsilon}{\varepsilon}
\newcommand{\e}{\mathbf e}
\newcommand{\wkl}{{\mf w_{\sf{GKSL}}}}

\newcommand{\tge}{\trianglerighteq}
\newcommand{\me}{\mathrm{e}}

\title{\LARGE \bf Provably Time-Optimal Cooling of Markovian Quantum Systems}

\author{Emanuel Malvetti
\thanks{E. Malvetti is with the School of Natural Sciences, Technische Universität München, Garching,
85737, Germany,
the Munich Center for Quantum Science and Technology (MCQST) \&
the Munich Quantum Valley (MQV).}
}

\begin{document}

\maketitle
\thispagestyle{empty}

\begin{abstract}
We address the problem of cooling a Markovian quantum system to a pure state in the shortest amount of time possible. 
Here the system drift takes the form of a Lindblad master equation and we assume fast unitary control.
This setting allows for a natural reduction of the control system to the eigenvalues of the state density matrix.
We give a simple necessary and sufficient characterization of systems which are (asymptotically) coolable and present a powerful result which allows to considerably simplify the search for optimal cooling solutions.
With these tools at our disposal we derive explicit provably time-optimal cooling protocols for rank one qubit systems, inverted $\Lambda$-systems on a qutrit, and a certain system consisting of two coupled qubits.
\end{abstract}



\section{INTRODUCTION}

Cooling quantum mechanical systems to a well-defined ground state is an essential task in quantum information technologies such as quantum computing~\cite{VincCriteria}.
For trapped ions or cold atoms, the preferred method is laser cooling, such as Doppler cooling, Sisyphus cooling~\cite{Wineland92}, Raman cooling~\cite{Monroe95b} and many others, as well as using strong coupling~\cite{Machnes10} or aided by optimal control~\cite{XQLi21}.
Another popular approach is using algorithmic cooling~\cite{SMW05,Boykin02,Popp06,Alhambra19}.
In this paper we use reduced control systems~\cite{Reduced23,LindbladReduced23} together with quantum optimal control theory~\cite{dAless21,DiHeGAMM08} to derive {\em provably time-optimal} schemes for cooling Markovian quantum systems.
Such Markovian systems are described by a time-independent master equation of \textsc{gks}--Lindblad form~\cite{GKS76,Lindblad76}. 
Furthermore, in the systems of concern, we assume that unitary control is fast compared to dissipation.
Corresponding results can be obtained assuming that the noise itself is switchable as in the experimental set-up of~\cite{Mart14,McDermott_TunDissip_2019}.

The main tool used here is going from
a full bilinear control system~\cite{Jurdjevic97,Elliott09} to a reduced one just
describing the dynamics of the eigenvalues of the state.
Such a 
reduced control system in a 
Lindbladian setting has first been formulated in~\cite{Sklarz04}, in~\cite{Yuan10} for a single qubit, and it has been studied in~\cite{rooney2018}. Similar ideas were 
used in the single-qubit setting in~\cite{Lapert10}.
Beyond single qubits, these tools (see also~\cite{CDC19,OSID23,vE_PhD_2020,MTNS2020_1}) have been vastly generalized and made mathematically rigorous in~\cite{Reduced23} and further explored in the Markovian setting in~\cite{LindbladReduced23} and applied to quantum entanglement in~\cite{BipartiteReduced24}. 
The reduced control system shifts the viewpoint naturally to the achievable derivatives of the eigenvalues and thus
invites the use of methods from the theory of {\em differential inclusions}.

\subsection*{Outline}

After sketching theoretical tools we focus on applications to illustrative low-dimensional examples in the language of a quantum engineer.
Yet the methods are general, readily carry over to higher dimensions and match with numerical methods.
Notably our time-optimal controls are obtained {\em without} invoking the Hamilton--Jacobi--Bellman Equation or the Pontryagin Maximum Principle. ---
The paper is organised as follows: 
Sec.~\ref{sec:reduced} recalls the definition of the reduced control system while
Sec.~\ref{sec:coolable} characterizes coolable systems in simple algebraic terms.
The methods are illustrated in Sec.~\ref{sec:qubit} by solving the rank-one qubit case is detail.
Sec.~\ref{sec:derivs} goes on to study achievable derivatives in higher dimensions and introduces two systems to be solved in the subsequent sections.
Sec.~\ref{sec:majorization} presents a method of reducing the achievable derivatives to a subset of optimal derivatives, and finally
Sec.~\ref{sec:opt-cooling} determines optimal solutions to the higher dimensional system introduced in Sec.~\ref{sec:derivs}.

\section{PRELIMINARIES}

Throughout the paper we work on the finite dimensional Hilbert space $\C^n$, and we denote the set of skew-Hermitian matrices using the unitary Lie algebra $\mf{u}(n)$.
Quantum states are represented by density matrices, that is, positive semi-definite matrices of unit trace, denoted $\mf{pos}_1(n)$.

\subsection{Full control system}

Markovian quantum systems are characterized by the \emph{\textsc{gks}--Lindblad equation}~\cite{GKS76,Lindblad76} which takes the form 
$$
\textstyle \dot\rho=-L(\rho) = -\iu[H_0,\rho]-\sum_{k=1}^r\Gamma_{V_k}(\rho),
$$ 
where
$\textstyle -\Gamma_V(\rho)=V\rho V^* -\tfrac12(V^*V\rho+\rho V^*V)$.
The \emph{Hamiltonian} $H_0\in\iu\mf{u}(n)$ is a Hermitian matrix and the \emph{Lindblad terms} $\{V_k\}_{k=1}^r\subset\C^{n,n}$ are arbitrary matrices. 
We call $-L$ the \emph{Kossakowski--Lindblad generator}\footnote{The signs are chosen such that the real parts of the eigenvalues of $-L$ are non-positive.}, and we denote the set of all Kossakowski--Lindblad generators in $n$-dimensions by $\wkl(n)$, called the \emph{Kossakowski--Lindblad Lie wedge}, cf.~\cite{DHKS08}.

The following definition encapsulates what we mean by a Markovian quantum system with fast unitary control.
Let $\{H_j\}_{j=1}^m$ be a set of Hermitian matrices, called \emph{control Hamiltonians}, and $I$ an interval of the form $[0,T]$ or $[0,\infty)$.
In the following we use the notation $\ad_H(\cdot)=[H,\cdot\,]$.
A path $\rho:I\to\mf{pos}_1(n)$ of density matrices is a solution of the full bilinear control system~\cite{Jurdjevic97,Elliott09}
\begin{align} \label{eq:bilinear} \tag{\sf F}
\textstyle \dot\rho(t) = 
-\big(\iu \sum_{j=1}^m u_j(t) \ad_{H_j}+L\big)(\rho(t)), \quad
\end{align}
with initial state $\rho(0)=\rho_0\in\mf{pos}_1(n)$ and with locally integrable control functions $u_j:I\to\R$ if $\rho$ is absolutely continuous and satisfies~\eqref{eq:bilinear} almost everywhere. 
We will always assume that the control Hamiltonians generate at least the special unitary Lie algebra: 
$$\generate{\iu H_j:j=1,\ldots,m}{Lie}\supseteq\mf{su}(n).$$
Since the full bilinear control system~\eqref{eq:bilinear} allows for unbounded control functions and since the control Hamiltonians generate the entire special unitary Lie algebra---meaning that we have fast unitary control---we can move arbitrarily quickly within the unitary orbits. Thus we may concentrate on the dynamics of the eigenvalues of the state.

\subsection{Definition of cooling} \label{sec:cost}

Since we assume to have fast unitary control over the system, any pure state can be transformed into any other pure state at no cost.
Hence any pure state can also be transformed into the ground state of the system Hamiltonian.
For this reason we will equate cooling the system with purifying it.
Moreover, due to the exponential nature of the \textsc{gks}--Lindblad equation pure states can only be reached asymptotically.
Thus one has to clarify what is meant by ``cooling the system in the shortest amount of time possible''.
In particular one has to define a cost (resp.\ reward) function.
The task then becomes to minimize (resp.\ maximize) this measure in a given time, or to reach a certain value in the shortest possible amount of time.

Some common examples of such measures are the purity of the state, the von Neumann entropy, the largest eigenvalue (which is the maximum fidelity with a pure state) or the minimum energy (with respect to some Hamiltonian with non-degenerate ground state). 
An important property of these functions is that they are Schur-convex (or concave), i.e.\ they are monotone with respect to majorization. The details are given in Sec.~\ref{sec:majorization}.

\section{REDUCED CONTROL SYSTEM} \label{sec:reduced}

The main tool for determining optimal cooling procedures used in this paper is the reduction of the full control system~\eqref{eq:bilinear} to a reduced control system describing the evolution of the eigenvalues of the density matrix $\rho$ representing the quantum state.
The general theory was established in~\cite{Reduced23} and applied in~\cite{LindbladReduced23} to investigate reachable and stabilizable states in Markovian quantum systems.

Since we have fast control over the unitary group, and since two density matrices have the same spectrum if and only if they lie on the same unitary orbit, we may concentrate on the dynamics of the eigenvalues of the state. Thus the reduced state space will be the standard simplex
\begin{align*}
\Delta^{n-1}=\big\{(\lambda_1,\ldots,\lambda_n)\in\R^n : \textstyle\sum_{i=1}^n \lambda_i=1,\, \lambda_i\geq0\,\; \forall i\big\}\,,
\end{align*}
representing the subset of diagonal density matrices.
More precisely a point $\lambda\in\Delta^{n-1}$ defines a unique diagonal density matrix $\rho=\diag(\lambda)\in\mf{pos}_1(n)$.
For the reverse direction, it is convenient to choose an ordering of the eigenvalues.
Let $\Delta^{n-1}_\down$ denote subset of $\Delta^{n-1}$ where the components $\lambda_i$ are arranged in a non-increasing manner.
Then we can define the map $\spec^\down:\mf{pos}_1(n)\to\Delta^{n-1}_\down$ which maps a density matrix to the vector containing its eigenvalues in non-increasing order.
The vertices of $\Delta^{n-1}$ are the standard basis vectors $e_i$ for $i=1,\ldots,n$ and the center is $\e=(1,\ldots,1)/n$.

Now the dynamics of the reduced control system are defined on the simplex $\Delta^{n-1}$ using the \emph{induced vector fields} $-L_U$ which can be defined as
\begin{equation*}
-L_U=J(U)-\diag(J(U)^\top\e),
\end{equation*}
where $J(U)_{ij}=\sum_{k=1}^r|\langle i|U^*V_kU|j\rangle|^2$.
A function $\lambda:I\to\Delta^{n-1}$ is a solution of the \emph{reduced control system}
\begin{align}
\label{eq:reduced} \tag{\sf R}
\dot\lambda(t) = -L_{U(t)} \lambda(t)\,,
\quad \lambda(0)=\lambda_0\in\Delta^{n-1}
\end{align}
with measurable control function $U:I\to\SU(n)$, if $\lambda$ is absolutely continuous and satisfies~\eqref{eq:reduced} almost everywhere. 

\begin{remark}
In~\cite{LindbladReduced23} additional definitions of the reduced control system are given.
Due to Filippov's Theorem \cite[Thm.~2.3]{Smirnov02} the system can equivalently be defined as a differential inclusion $\dot\lambda(t) \in \{-L_U\lambda(t):U\in\SU(n)\}$, which is a more geometric way of describing the system.
Moreover one may take the convex hull of the set of achievable derivatives, as solutions to this \emph{relaxed control system}
\begin{equation} \label{eq:relaxed} \tag{\sf C}
\dot\lambda(t) \in \mf \conv(-L_U\lambda(t))\,,
\quad \lambda(0)=\lambda_0\in\Delta^{n-1}
\end{equation}
can still be uniformly approximated on compact time intervals by solutions to~\eqref{eq:reduced}.
This is the content of the Relaxation Theorem~\cite[Ch.~2.4, Thm.~2]{Aubin84}.
\end{remark}

The reduced control system~\eqref{eq:reduced} is indeed equivalent to the full control system~\eqref{eq:bilinear} in the precise sense of the Equivalence Theorem~\cite[Thm.~2.4]{LindbladReduced23}. 
In particular no loss of information is incurred by switching to the reduced control system.
A convenient way to state the equivalence is through reachable sets.

First we recall the definition of reachable sets for~\eqref{eq:bilinear}. The definitions for~\eqref{eq:reduced} are entirely analogous. The \emph{reachable set of $\rho_0$ at time $T\geq0$} is defined as
$$
\reach_{\ref{eq:bilinear}}(\rho_0,T) =
\{\rho(T):\rho \text{ solves } \eqref{eq:bilinear},\, \rho(0)=\rho_0\}.
$$
Similarly the \emph{all-time reachable set of $\rho_0$} is defined as
$\reach_{\ref{eq:bilinear}}(\rho_0) =\bigcup_{T\geq0}\reach_{\ref{eq:bilinear}}(\rho_0,T)$.
Then we can state the result as:

\begin{proposition} \label{prop:reach-equiv}
Given any $-L\in\wkl(n)$, 
let $\rho_0\in\mf{pos}_1(n)$ and $\lambda_0\in\Delta^{n-1}$ be such that $\spec^\down(\rho_0)=\lambda^\down_0$.
Then it holds for all $T>0$ that
\begin{align*}
\overline{\reach_{\ref{eq:bilinear}}(\rho_0,T)}
= 
\overline{\{U\lambda U^*:\lambda\in\reach_{\ref{eq:reduced}}(\lambda_0,T),U\in\SU(n)\}}.
\end{align*}
\end{proposition}

\smallskip 

The analogous result also holds for the all time reachable sets.
A general proof is provided in~\cite{Reduced23}.

To be of practical use one needs a way to lift solutions from the reduced control system to the full one.
In particular one needs a method to determine corresponding control functions.
A general result is given in~\cite[Prop.~3.10]{Reduced23}, and we will apply it to some concrete examples in Sections~\ref{sec:qubit} and~\ref{sec:opt-cooling}.

\section{ASYMPTOTICALLY COOLABLE SYSTEMS} \label{sec:coolable}

The first question to ask is whether the system under consideration is coolable at all, by which we mean that a pure quantum state state can be reached from every given initial state, at least in an asymptotic sense.
It turns out that such systems can be characterized nicely in an algebraic way, cf.~\cite[Thm.~4.7]{LindbladReduced23}.

\begin{theorem} \label{thm:asymptotic-coolability}
Given any Kossakowski--Lindblad generator $-L\in\wkl(n)$,
the following are equivalent.\smallskip
\begin{enumerate}[(i)]
\item \label{it:cool-evec} For each choice of Lindblad terms $\{V_k\}_{k=1}^r$ of $-L$,
there exists a common eigenvector of all $V_k$ which is not a common left eigenvector.
\item \label{it:cool-ham} There exists a (time-independent) Hamiltonian $H$ such that $-(\iu\ad_H + L)$ has a (unique) attractive fixed point\footnote{We say that $\rho$ is an attractive fixed point if every solution converges to $\rho$. If such an attractive fixed point exists, it is clearly unique.}, and this fixed point is pure.
\item \label{it:cool-conv} For every initial state, there exists some solution to~\eqref{eq:reduced} converging to $e_1$.
\item \label{it:cool-reach-some} There exists $\lambda\in\Delta^{n-1}\setminus\{e_1\}$ such that $e_1\in\overline{\reach_{\ref{eq:reduced}}(\lambda)}$.
\end{enumerate}
\end{theorem}

An efficient algorithmic way of determining whether a set of matrices $V_k$ has a common eigenvector, and, if so, of finding such an eigenvector, is presented in~\cite{Triangulation23}.

In the following we will only deal with asymptotically coolable systems. If the system is not coolable, one first has to determine which states are reachable, and which of them is the coolest by some appropriate measure.
Several results in this direction can be found in~\cite{OSID23,LindbladReduced23}.

\section{\/`WARM-UP\/': OPTIMAL COOLING OF A QUBIT} \label{sec:qubit}

The simplest special case is that of a single qubit. 
In this section we focus on rank one systems, which are defined by a single Lindblad term $V$.
These systems are not trivial, but they still allow for a complete description.
The general qubit case (including non-coolable systems) is treated in~\cite{QubitReduced24}.
The solution obtained here shares some similarity with the solution obtained for the special case of the Bloch equations~\cite{Lapert10,QubitReduced24}, which is however simpler to solve as it has a symmetry which allows to reduce the dimensionality of the problem.

For later use we define the Pauli matrices 
$$
\sigma_x=\begin{pmatrix}0&1\\1&0\end{pmatrix},\quad
\sigma_y=\begin{pmatrix}0&-\iu\\\iu&0\end{pmatrix},\quad
\sigma_z=\begin{pmatrix}1&0\\0&-1\end{pmatrix}.
$$ 

Let $V\in\C^{2,2}$ be an arbitrary Lindblad term.
First we have to check whether such a system is coolable at all.
Indeed, Theorem~\ref{thm:asymptotic-coolability} shows that the system is asymptotically coolable if and only if $V$ is not normal (equivalently if and only if $-L$ is not unital).

At a first glance there are eight real parameters defining the problem but the following result shows that all but one parameter can be eliminated.

\begin{lemma} \label{lemma:reduced-parameters}
Let $V\in\C^{2,2}$ be an arbitrary non-normal matrix.
Then there exists a (special) unitary matrix $\tilde U$, a Hermitian matrix $\tilde H$ and numbers $\gamma>0$ and $\nu\in[0,1)$ with 
$$
\tilde H = \tfrac\iu4(\tr(V^*)V-\tr(V)V^*), 
\quad 
\tilde V = \begin{pmatrix}0&1\\\nu&0\end{pmatrix},
$$
such that $\Gamma_V = \iu\ad_{\tilde H} + \gamma\Gamma_{\tilde U^*\tilde V\tilde U}$. 
\end{lemma}

\begin{proof}
Let $\tilde U$ be a unitary such that $\tilde U[V,V^*]\tilde U^*$ is diagonal. 
Then it is easy to show that $\tilde U(V-\tr(V)/2)\tilde U^*$ is zero on the diagonal.
By adjusting $\tilde U$ (without renaming) we can additionally make sure that the off-diagonal elements have the same argument and the top right element has the greater modulus.
Together this gives $V-\tr(V)/2=\sqrt\gamma e^{\iu\phi}\tilde U^*\tilde V\tilde U^*$.
The freedom of representation of the Lindblad equation~\cite[Lem.~C.3]{LindbladReduced23} then implies
$
\Gamma_V 
= 
\iu\ad_{\tilde H} + \Gamma_{V-\tr(V)/2}
=
\iu\ad_{\tilde H} + \gamma\Gamma_{\tilde U^*\tilde V\tilde U}
$
as desired.
\end{proof}

There are two extremal cases. 
If $\nu=0$ we obtain a special case of the Bloch equations, and if $\nu=1$, the matrix $V$ is normal and hence the system is not coolable.
For this reason we exclude the case $\nu=1$.

\subsection{Space of generators}

From now on we assume that we have a single Lindblad term of the form $\tilde V$ as in Lemma~\ref{lemma:reduced-parameters} depending only on the parameter $\nu\in[0,1)$. The general solution will then be recovered at the end. All figures in this section use the value $\nu=1/2$.

In the qubit case the generators $-L_U$ are defined by two non-negative real numbers on the off diagonal and hence they can be easily visualized.
Indeed, the main tool in the following will be the \emph{space of generators} $\mf Q$ which is linearly isomorphic to the set of all $-L_U$:
$$
\mf Q=
\{(J_{12}(U)-J_{21}(U),J_{12}(U)+J_{21}(U)) : U\in\SU(2)\}. 
$$
For rank one systems ${\mf Q}$ takes on a rather simple form, see Figure~\ref{fig:rank-one-gens}.

\begin{figure}[h]
\centering
\includegraphics[width=0.40\textwidth]{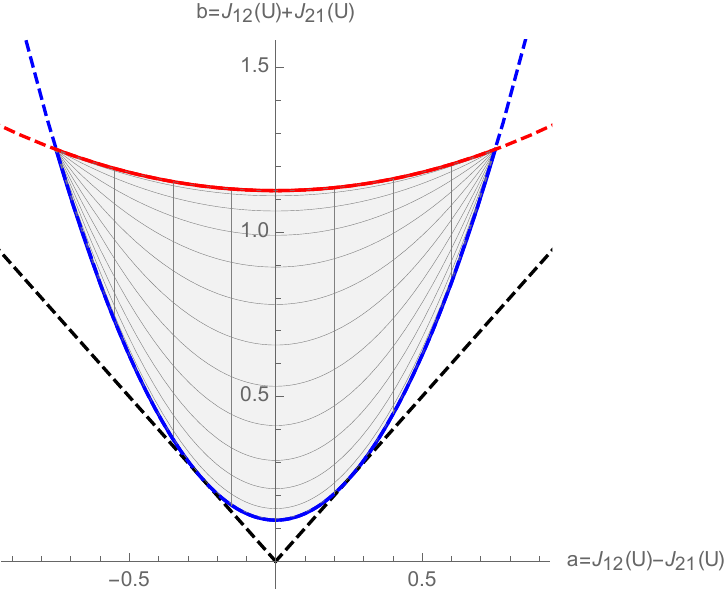}
\caption{The parametrized space of generators $\mf Q$ of a rank one system as given in Lemma~\ref{lemma:rank-one-space-of-lines} and Corollary~\ref{coro:rank-one-space-of-lines}. 
We work in a basis where $[V,V^*]$ is diagonal. 
The poles are mapped to the corners $(\pm(1-\nu^2),1+\nu^2)$. 
The latitude lines are vertical, and the equator lies on the $y$-axis. The longitude lines are parabolas passing through the poles and intersecting the $y$-axis between $b=\frac12(1\pm\nu)^2$.}
\label{fig:rank-one-gens}
\end{figure}

\begin{lemma} \label{lemma:rank-one-space-of-lines}
The space of generators can be parametrized as 
$
{\mf Q} = \{(a,f_z(a)) : a\in[-(1-\nu^2),1-\nu^2], \,z\in[0,1) \},
$ 
where
\begin{align*} 
f_z(a) = 1+\nu^2 - 
\frac{1+\nu^2-2\cos(4\pi z)\nu}{2}
\Big(1-\Big(\frac{a}{1-\nu^2}\Big)^2\Big), 
\end{align*}
In fact, the point $(a,f_z(a))$ can be obtained using the unitary $U=\me^{\iu\pi z\sigma_z}\me^{\iu\pi x\sigma_x}$ where $x$ satisfies $1-a/(1-\nu^2)=2\sin(\pi x)^2$ and $z=1/4$.
\end{lemma}

\begin{proof}
With $U$ as above we find $(a,b)\in{\mf Q}$ with
$a=(1-2r)(1-\nu^2)$, and
$b=1+\nu^2+2r(r-1)|1-e^{4\iu\pi z}\nu|^2$,
where $r=\sin(\pi x)^2$.
Hence with some basic trigonometry
we get 
$b = \Sigma - \tfrac12|1-e^{4\iu\pi z}\nu|^2
\big(1-\tfrac{a^2}{(1-\nu^2)^2}\big)$.
\end{proof}

It is of particular importance to understand the boundary of this set. 

\begin{corollary} \label{coro:rank-one-space-of-lines}
For $V$ non-normal, the space of generators is the region enclosed between the two parabolic segments 
\begin{align*}
a\mapsto 1+\nu^2 - 
\tfrac12(1\pm\nu)^2
\big(1-\big(\tfrac{a}{1-\nu^2}\big)^2\big),
\end{align*}
on $a\in[-1+\nu^2,1-\nu^2]$.
\end{corollary}

\begin{remark}
The two extremal parabolas of Corollary~\ref{coro:rank-one-space-of-lines} (considered on $\R$) are the unique parabolas passing through the points $(\pm(1-\nu^2),1+\nu^2)$ which are tangent to the lines $a\mapsto\pm a$. 
The points of tangency in ${\mf Q}$ are achieved when $U^*\tilde VU$ is upper or lower triangular.
Note that for $\nu=0$ the two parabolas coincide (which is consistent with the Bloch case), and as $\nu\to1$ all points tend to the $y$-axis which is consistent with unital systems, cf.~\cite{QubitReduced24}.
\end{remark}

\subsection{Optimal derivatives and path}

In the qubit case it is convenient to represent the reduced state by the first eigenvalue $\lambda\in[0,1]$.
The maximal achievable derivative of $\lambda$, denoted by $\mu:[0,1]\to\R$, can be obtained from the boundary of ${\mf Q}$ via the relation
$$
\mu(\lambda)=\max_{(a,b)\in{\mf Q}}\tfrac{1}{2} (a+(1-2\lambda)b).
$$
Details are given in~\cite{QubitReduced24}.
A computation then yields the following result, see also Figure~\ref{fig:rank-one-mu} for an illustration.

\begin{lemma} \label{lemma:mu}
Let $V$ be non-normal and $\lambda_0=\tfrac12(1+\tfrac{1-\nu}{1+\nu})$, then the maximal derivative $\mu:[0,1]\to\R$ takes the form
\begin{align*} 
\mu(\lambda) =
\begin{cases}
\tfrac12(1-\nu^2 - (1+\nu^2)(2\lambda-1))
&\text{ if } 0\leq\lambda\leq\lambda_0\\[1mm]
(\tfrac{1-\nu}{2})^2 (\frac{1}{2\lambda-1}+2\lambda-1)
&\text{ if } \lambda_0\leq\lambda\leq1.
\end{cases}
\end{align*}
The function $\mu$ is continuously differentiable.
\end{lemma}

\begin{figure}[h]
\centering
\includegraphics[width=0.40\textwidth]{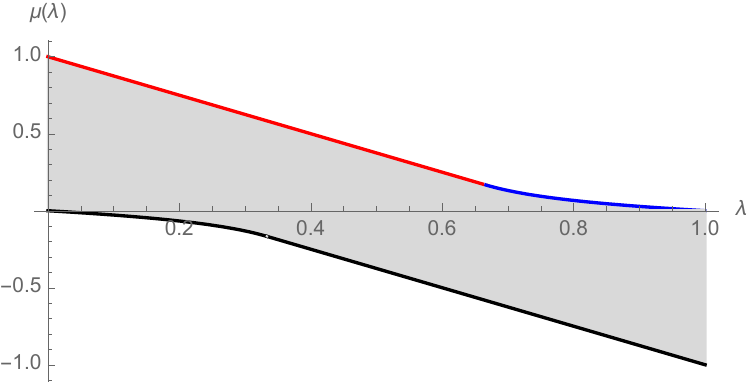}
\caption{Achievable derivatives as a set-valued function of $\lambda$ with the upper bound $\mu$, given in Lemma~\ref{lemma:mu}, highlighted.}
\label{fig:rank-one-mu}
\end{figure}

From this the optimal path through the Bloch ball can be computed.

\begin{lemma} \label{lemma:opt-path}
Let $t_0=\log\big(1-\frac{1+\nu^2}{1+\nu}\big)/(1+\nu^2)$.
The optimal path through the Bloch ball is given by $\rho^\star(t)=\Ad_{U^\star(t)}(\diag(\lambda^\star(t)))$ with $U^\star(t)=e^{\iu\pi y^\star(t)\sigma_y}$ and where
$$
\lambda^\star(t)=
\begin{cases}
\frac{1-e^{-(1+\nu^2)t}}{1+\nu^2} \quad&0\leq t\leq t_0\\[1mm]
\frac12\big( 1 + \sqrt{1-ce^{-(1-\nu)^2 t}} \big) \quad&t\geq t_0
\end{cases}
$$
with $c=\frac{4\nu}{(1+\nu^2)}\big(\frac{1+\nu}{\nu(1-\nu)}\big)^{\frac{(1-\nu)^2}{(1+\nu)^2}}$, is the unique solution to $\tfrac{d}{dt}\lambda^\star(t)=\mu(\lambda^\star(t))$ with 
$\lambda^\star(0)=0$ and the (continuous) function $y^\star$ is defined by
$$
y^\star(t)=\begin{cases}
0 \quad&0\leq t\leq t_0\\
\tfrac1\pi\arcsin\big(\sqrt{
\tfrac12(1+\frac{1-\nu}{1+\nu}\frac{1}{1-2\lambda^\star(t)})
}\big) \quad&t\geq t_0.
\end{cases}
$$
\end{lemma}
\smallskip
\begin{proof}
The function $\lambda^\star$ is found by integration and can be checked by differentiating, and $y^\star$ as a function of $\lambda^\star$ is found by computing the values of $x$ and $z$ which give the boundary point of ${\mf Q}$ corresponding to $\lambda$.
\end{proof}

The optimal path is illustrated in Figure~\ref{fig:rank-one-path}. 
Note that the path is much simpler than one might expect from the formulas.

\begin{figure}[h]
\centering
\includegraphics[trim={0 1cm  0 0},clip,width=0.40\textwidth]{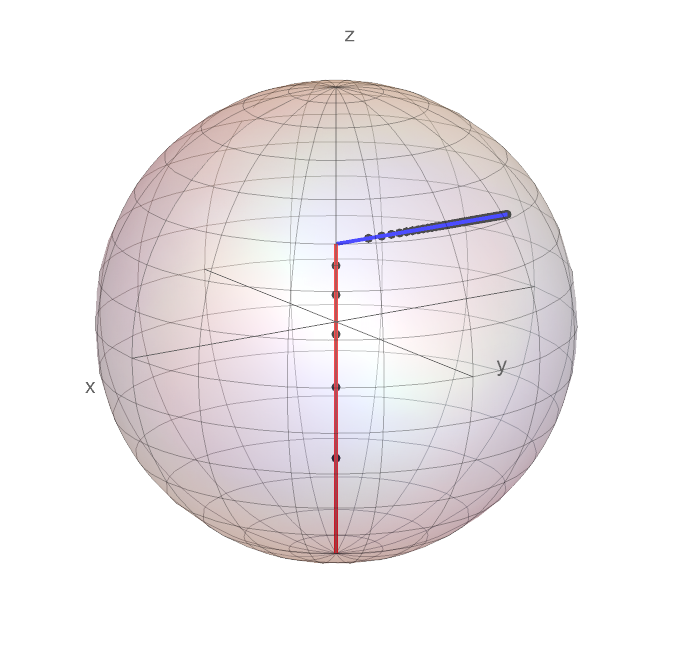}
\caption{Time-optimal path from the boundary of the Bloch ball (pure state) to the center (maximally mixed state) and back. Starting at the south pole, the path follows the $z$-axis until $\lambda=\lambda_0$. Then the path takes a sharp turn and continues horizontally until it reaches the boundary (which happens only asymptotically). 
When projected onto the $x,y$-plane, the horizontal part is a straight line lying on the negative $x$-axis (the mirrored path along the positive $x$-axis is also optimal).
The black dots are equally spaced in time, and accumulate towards the end.
The solution shares some similarity with the so-called magic plane result for the Bloch equations obtained in~\cite{Lapert10}.}
\label{fig:rank-one-path}
\end{figure}

\subsection{Optimal controls}

So far we have found the optimal derivatives of $\lambda$ and the optimal path of $\rho$ through the Bloch ball.
It remains to determine the corresponding optimal controls of the full control system~\eqref{eq:bilinear}.
To simplify the problem we assume that the control Hamiltonians are the Pauli matrices and the goal is to determine the corresponding control functions $u_x,u_y$ and $u_z$.

There are two contributions to the control Hamiltonians.
A \emph{direct} term obtained by differentiating the optimal control unitary of the reduced system and a \emph{compensating} term which cancels out the motion tangent to the unitary orbits induced by the drift $-L$, see~\cite[Prop.~3.10]{Reduced23}.
This leads to the following result:

\begin{proposition}
A choice of optimal controls is given by $u_x=u_z\equiv0$ on $[0,\infty)$. Moreover $u_y\equiv0$ on $[0,t_0]$ and
\begin{align*}
u_y(t) 
=&
\medmath{
\frac{\frac1{2\lambda-1} \frac{1-\nu}{1+\nu}}
{\sqrt{(2\lambda-1)^2-(\frac{1-\nu}{1+\nu})^2}} 
}
\cdot
\medmath{\Big(\frac{1-\nu}2\Big)^2 \Big(\frac1{2\lambda-1}-(2\lambda-1)\Big)}
\\ &+ 
\medmath{\frac{(1+\nu)\sin(2\pi y)}4
\Big(\frac{2(1-\nu)}{2\lambda-1} + (1+\nu)\cos(2\pi y)\Big)}\,,
\end{align*}
on $t\in[t_0,\infty)$ and where $\lambda^\star(t)$ is as in Lemma~\ref{lemma:opt-path}.
\end{proposition}

\begin{proof}
The direct term is $\big(\frac{d}{dt}U_{y^\star(t)}\big)U^*_{x^\star(t)}$ where $U_x=e^{\iu\pi x\sigma_x}$ and $x^\star(t)$ is as in Lemma~\ref{lemma:opt-path}.
Using the chain rule this becomes $\iu\pi\sigma_x\frac{dx^\star}{d\lambda^\star}(\lambda^\star(t))\mu(\lambda^\star(t))$.
The expression follows directly from the evaluation of the derivative.
The compensation term takes the form 
$\ad_\rho^+(L(\rho))$, cf.~\cite[Prop.~3.10]{Reduced23},
where 
$\ad_\rho(\cdot)=[\rho,\cdot]$ and $(\cdot)^+$ is the Moore--Penrose inverse. 
The result then follows from an elementary computation.
\end{proof}

The obtained optimal controls are illustrated in Figure~\ref{fig:rank-one-ctrl}.

\begin{figure}[h]
\centering
\includegraphics[width=0.40\textwidth]{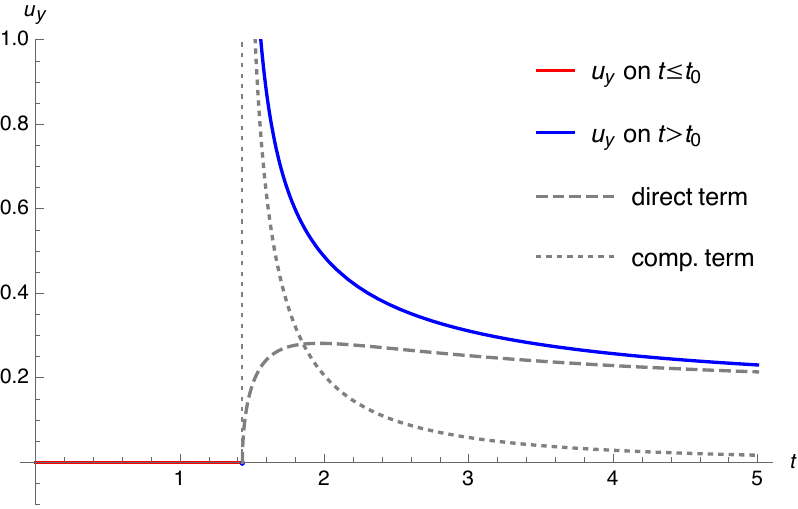}
\caption{Optimal control function $u_y$ (solid) with direct (dashed) and compensation (dotted) contributions using $\nu=1/2$. 
The control is identically zero on $[0,t_0]$ and has a singularity at $t_0$.}
\label{fig:rank-one-ctrl}
\end{figure}

\begin{remark}
For general non-normal $V$ and with initial state is $\ketbra0$ the optimal controls are the following, using the notation from Lemma~\ref{lemma:reduced-parameters}.
First (almost) instantaneously apply the unitary $\tilde U$ to the system.
Then apply the control Hamiltonian
$$
H(t)=-\tilde H 
+ \gamma u_{\nu,y}(\gamma t) \, \tilde U\sigma_y\tilde U^*
$$
for $t\in[0,\infty)$.
\end{remark}

\section{ACHIEVABLE DERIVATIVES} \label{sec:derivs}

Locally the reduced control system can be understood by studying the set of \emph{achievable derivatives} at $\lambda$, denoted $\derv(\lambda):=\{-L_U\lambda:U\in\SU(n)\}$.
Due to continuity of the map $U\mapsto -L_U\lambda$ it is clear that the set is compact and path-connected, but the exact shape is difficult to determine in general.

In the qubit case studied above, $\derv(\lambda)$ was just a closed interval, and we were able to give an analytical expression.
For more general qubit systems this task becomes more difficult, but it still allows for a partial analytical solution~\cite{QubitReduced24}.
In higher dimensional systems the shape of $\derv(\lambda)$ can be quite arbitrary, but in some special cases it takes the form of a (convex) polytope, which can be seen as a generalization of the qubit case.
In general however this is not true, although it is still useful to approximate $\derv(\lambda)$ with polytopes, both from the inside and the outside.

In the remainder of this section we provide some results and examples in this direction.
Note that due to the Relaxation Theorem~\cite[Ch.~2.4, Thm.~2]{Aubin84} we are also interested in the convex hull of $\derv(\lambda)$.

\subsection{Examples of polytopes}

We present a few cases where the set of achievable derivatives $\derv(\lambda)$ takes the form of a polytope, and some examples where $\derv(\lambda)$ is not convex and $\conv(\derv(\lambda))$ is not a polytope.

First, at the maximally mixed state $\e/n$, the set $\derv(\e/n)$ is always a convex polytope. In fact the vertices are the vectors containing the eigenvalues of $\sum_{k=1}^r[V_k,V_k^*]$ in all possible permutations. This can be shown using~\cite[Lem.~B.2]{LindbladReduced23} together with the Schur--Horn Theorem~\cite{Schur23,Horn54}.

Another special case occurs at the vertices $e_i$ of the simplex $\Delta^{n-1}$ under the assumption that there is only one Lindblad term $V$.
Indeed there exists a value $0<f^\star\leq \|V\|_\infty$, where $\|\cdot\|_\infty$ is the Schatten $\infty$-norm (i.e.\ the largest singular value), such that $$\derv(e_i)=f^\star\conv(\{0,e_j-e_i:j\neq i\}).$$
Let $f(U)$ be the sum of the squares of the off-diagonal elements in the first column of $J(U)$.
Then $f^\star$ is the maximal value of $f(U)$, and the value $0$ is achievable using the Schur decomposition.
Using unitaries which leave the first basis vector invariant the entire polytope can be obtained. 

It is important to note that generically $\derv(\lambda)$ is not a polytope.
For a simple counterexample consider the Lindblad term $V=\ketbrax12+\sqrt2\ketbrax23$ on a qutrit.
A numerical computation shows that for general $\lambda$ the set $\derv(\lambda)$ is not convex.

In the remainder of this section we introduce two concrete systems where $\derv(\lambda)$ is always a polytope, and they will serve as running examples in the following sections.

\subsubsection{Qutrit systems with spontaneous emission}

Since the qubit case is addressed in detail in~\cite{QubitReduced24}, the next logical step is the qutrit.
We consider the special form of system with only spontaneous emissions as described in~\cite{Sklarz04} and which include the $\Lambda$-system.
Such systems are defined by Lindblad terms of the form $\sqrt{\gamma_{ij}}\ketbrax{i}{j}$ for $i,j\in\{1,\ldots,n\}$.
These systems have the property that
$J(U)=\Theta^\top\Gamma\Theta$ where $\Gamma=J(\id)$ and $\Theta$ is the unistochastic matrix defined by $\Theta_{ij}=|U_{ij}|^2$.
We focus on the following two systems:
$$
\Gamma^{\Lambda}=\begin{pmatrix}
0&0&0\\
\gamma_1&0&0\\
\gamma_2&0&0
\end{pmatrix}
,\quad
\Gamma^{\mathrm V}=\begin{pmatrix}
0&\gamma_1&\gamma_2\\
0&0&0\\
0&0&0
\end{pmatrix}.
$$
The first one is the $\Lambda$-system studied in~\cite{Sklarz04} and the second one is an inverted version, which we call the $\mathrm V$-system.
We now show that for a $\mathrm V$-system $\derv(\lambda)$ is always a convex polytope.

\begin{proposition} \label{prop:V-sys}
For the $\mathrm V$-system in three dimensions it holds that
$\derv(\lambda) = \conv(\{-L_P\lambda:P\in S_3\})$,
where $S_3$ is the symmetric group represented by permutation matrices.
\end{proposition}

\begin{proof}
We find that 
the generator $-L_U$
equals the convex combination of generators $u_1\Gamma_1+u_2\Gamma_2+u_3\Gamma_3$ where the $\Gamma_i$ are
$$
\begin{pmatrix}
0&\tilde\gamma_2&\tilde\gamma_3\\
0&-\tilde\gamma_2&0\\
0&0&-\tilde\gamma_3
\end{pmatrix},
\begin{pmatrix}
-\tilde\gamma_1&0&0\\
\tilde\gamma_1&0&\tilde\gamma_3\\
0&0&-\tilde\gamma_3
\end{pmatrix},\begin{pmatrix}
-\tilde\gamma_1&0&0\\
0&-\tilde\gamma_2&0\\
\tilde\gamma_1&\tilde\gamma_2&0
\end{pmatrix}
$$
and where 
$(\tilde\gamma_1,\tilde\gamma_2,\tilde\gamma_3)=(\gamma_2,\gamma_1,0)\Theta$.

It suffices to show that $\Gamma_1\lambda$ lies in the desired polytope.
Consider the following linear bijection: 
$$
(\dot\lambda_1,\dot\lambda_2,\dot\lambda_3)
\mapsto
(x,y):=(-\dot\lambda_2/b,-\dot\lambda_3/c).
$$
In these coordinates, the polytope intersected with the first quadrant is defined by
$$
x,y\leq\max(\gamma_1,\gamma_2),\quad 
x+y\leq \gamma_1+\gamma_2
$$
and clearly $\Gamma_1\lambda$ satisfies this since $(\tilde\gamma_1,\tilde\gamma_2,\tilde\gamma_3)\preceq(\gamma_2,\gamma_1,0)$ as $\Theta$ is bistochastic.\footnote{The symbol $\preceq$ denotes majorization, see Section~\ref{sec:majorization} for details.}
\end{proof}

\begin{figure}[h]
\centering
\includegraphics[width=0.40\textwidth]{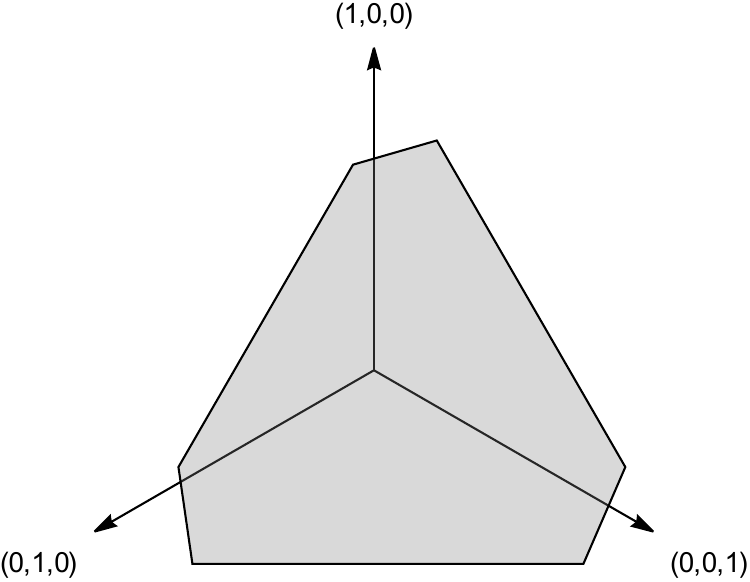}
\caption{Achievable derivatives in the $\mathrm V$-system with $\gamma_1=1$ and $\gamma_2=2$ at the point $\lambda=(0.4,0.35,0.25)$.}
\label{fig:inv-lambda-derivs}
\end{figure}

One might be tempted to try and generalize this result to all qutrit systems with spontaneous emission, but unfortunately it fails already for the $\Lambda$-system, see Figure~\ref{fig:lambda-derivs} and the following example.

\begin{example}
Consider the $\Lambda$-system with $\gamma_1=\gamma_2=1$ and let $\lambda=(1,0,0)^\top$.
Using only permutations we obtain the derivatives $(0,0,0)^\top$ and $(-2,1,1)^\top$.
However, a part of the boundary can be obtained by computing
$$
\begin{pmatrix}
x&y&0\\
y&x&0\\
0&0&1
\end{pmatrix}
\begin{pmatrix}
0&0&0\\
1&0&0\\
1&0&0
\end{pmatrix}
\begin{pmatrix}
x&y&0\\
y&x&0\\
0&0&0
\end{pmatrix} 
=
\begin{pmatrix}
*&y^2&0\\
xy&*&0\\
x&y&0
\end{pmatrix},
$$
where $x\in[0,1]$ and $y=1-x$. Clearly the resulting derivatives do not lie in the convex hull of $(0,0,0)^\top$ and $(-2,1,1)^\top$.
\end{example}

It is somewhat unexpected that the $\Lambda$-system does not have this polytope property, since the optimal cooling solution found in~\cite{Sklarz04} only requires permutations. 
This is explained by the fact that the ``non-classical'' achievable derivatives are suboptimal for cooling, cf.\ Sec.~\ref{sec:majorization}.

\begin{figure}[h]
\centering
\includegraphics[width=0.35\textwidth]{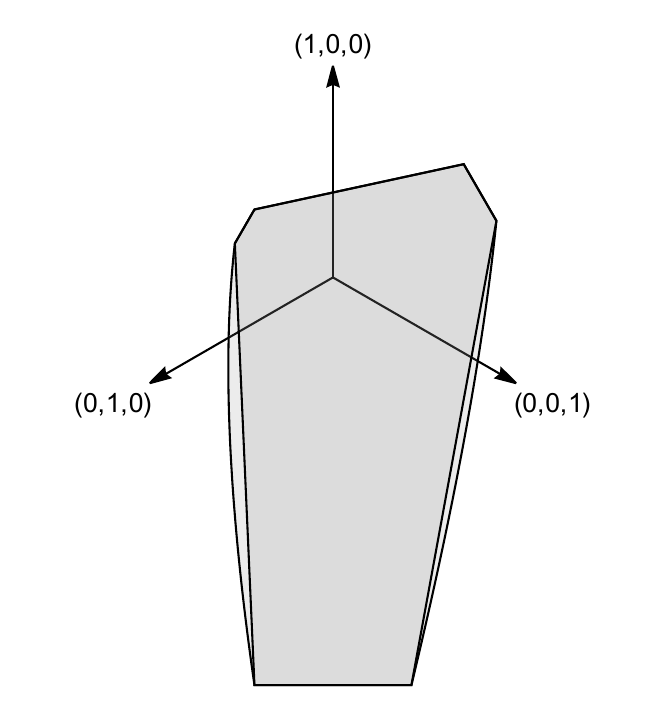}
\caption{Achievable derivatives in the $\Lambda$-system with $\gamma_1=1$ and $\gamma_2=2$ at the point $\lambda=(0.6,0.25,0.15)$. The corners are achieved by permutations, and there are bulges which leave the polytope but belong to $\derv(\lambda)$.}
\label{fig:lambda-derivs}
\end{figure}

\subsubsection{Spin-spin system}

In this section we explore a system composed of two qubits with a single Lindblad term $V=\sigma_-\otimes\id$ where $\sigma_-$ is the lowering operator.
This can be seen as a first approximation for the ubiquitous spin-boson system.
We will conjecture an exact description of the (convex hull) of the achievable derivatives and support it with a partial proof and numerical evidence.
In Sec.~\ref{sec:upper-bounds} we give a slightly larger upper bound with full proof and in Sec.~\ref{sec:opt-cooling} we derive an optimal cooling procedure for the system.

\begin{conjecture} \label{conj:spin-spin}
For every $\lambda\in\Delta^3$ it holds that $$\conv(\derv(\lambda))=\conv(\{-L_P\lambda:P\in S_4\})=:\mathcal P(\lambda).$$
\end{conjecture} \smallskip

The $\supseteq$ direction is trivially satisfied, so it remains to show that the right-hand side is an upper bound of $\derv(\lambda)$.
First we need to better understand the polytope $\mathcal P(\lambda)$.
Indeed we can derive a simple inequality description of the polytope:

\begin{lemma}
For regular $\lambda$, the polytope $\mathcal P(\lambda)$ has $12$ vertices and $8$ facets, $4$ of which are hexagonal and the other $4$ are triangular.
The hexagonal facets are described by $\dot\lambda_i\geq-\lambda_i$.
Let $a$ be any of the four eigenvalues and $b\geq c\geq d$ be the remaining ones.
Then 
$$
\dot a(b+d) - \dot b(c-b) - b(c+d) \leq 0
$$
describes the corresponding triangular facet.
\end{lemma}

The number of vertices stems from the fact that $V$ is invariant under the permutation $(12)(34)$.
The hexagonal facet inequalities are clearly satisfied, and so it remains to show that all achievable derivatives also satisfy the inequalities triangular facet.
So far we only have numerical evidence for this claim.
Yet 
a slightly larger provable bound can be obtained by studying the set of matrices $J(U)$, cf.\ \mbox{Sec.~\ref{sec:upper-bounds}}.

\subsection{Reduction to toy model}

We say that $-L$ is \emph{quasi-classical} if for all $\lambda\in\Delta^{n-1}$ it holds that $\derv(\lambda)\subseteq\conv(\{-L_P\lambda:P\in S_n\})$.

Hence by Proposition~\ref{prop:V-sys} the V-system on the qutrit is quasi-classical, but the qubit system studied in Sec.~\ref{sec:qubit} is not.
Such quasi-classical systems are particularly nice, because the controls in the reduced control system can be restricted to the permutations (up to convexification).

Systems with this limited set of controls have been studied in some detail in the setting of quantum thermodynamics~\cite{CDC19,OSID23,vE_PhD_2020,MTNS2020_1} under the name ``toy model''.
Note however that these systems are not quasi-classical, meaning that ``non-classical'' controls might improve the results obtained therein.

On the other hand these systems have the convenient property that diagonal density matrices remain diagonal, and hence no compensating Hamiltonian is necessary.
We will call such systems \emph{diagonally invariant}.
Indeed this follows from the formula of the compensating Hamiltonian given in~\cite[Sec.~3.2]{LindbladReduced23}. 
More generally a system is diagonally invariant if for all Lindblad terms $V_k$ the associated directed graph (with arcs corresponding to non-zero matrix entries) are disjoint unions of directed paths and directed cycles.
In particular systems with spontaneous emission and the spin-spin system introduced in the previous section have the property that diagonal states remain diagonal.

\subsection{Upper bounds and speed limits} \label{sec:upper-bounds}

Working with the toy model, i.e.\ restricting the controls of the reduced control system to permutations, is a practical way of simplifying the system and allows for the computation of solutions which are not necessarily optimal.
In effect, this method approximates $\conv(\derv(\lambda))$ from the inside.
Conversely, this section is concerned with approximations from the outside.
This is useful as it yields speed limits and upper bounds to optimal solutions.

First we look at some general results before considering specific systems.
Useful bounds can be obtained on the level of the $J(U)$ matrices.
A somewhat trivial bound can be obtained by setting $\gamma=\sum_{k=1}^r\|V_k\|_2^2$. Then all matrices $J=J(U)$ satisfy the inequalities $\textstyle\sum_{i,j=1}^n J_{ij} \leq \gamma$ and $J_{ij}\geq0$.
Using~\cite[Lem.~B.2]{LindbladReduced23} we get the following stronger bound.

\begin{lemma} \label{lemma:majorization-bounds}
Let $J=J(U)$, then $J_{ij}\geq0$ as well as 
$J\e\preceq\spec(\textstyle\sum_{k=1}^r V_kV_k^*)$, and
$J^\top\e\preceq\spec(\textstyle\sum_{k=1}^r V_k^*V_k)$, and
\begin{align*}
(J+J^\top)\e&\preceq\spec(\textstyle\sum_{k=1}^r \{V_k,V_k^*\}) \\
(J-J^\top)\e&\preceq\spec(\textstyle\sum_{k=1}^r [V_k,V_k^*]),
\end{align*}
and this defines a polytope bound for the set of all $J(U)$.
\end{lemma}

Note that by linearity these also define corresponding polytope bounds for each $\derv(\lambda)$.
An important property of these bounds is that they guarantee that the simplex $\Delta^{n-1}$ is preserved, and hence they do not lead to unphysical behavior.

For the spin-spin system we have, so far, only conjectured a polytope bound on $\derv(\lambda)$.
By studying the set of matrices $J(U)$, we can prove a slightly larger bound.

\begin{lemma} \label{lemma:spin-spin-upper-bound}
Consider the spin-spin system. For every matrix $J=J(U)$ it holds that
$$
J_{ij}\geq0, \quad (J+J^\top)\e=\e, \quad J_{ii}\leq\tfrac14
$$
for all $i,j\in\{1,\ldots,n\}$.
\end{lemma}

\begin{proof}
The first two constraints follow immediately from Lemma~\ref{lemma:majorization-bounds} and its proof (cf.~\cite[Lem.~B.2]{LindbladReduced23}), and the last one follows from $J(U)_{ii}=|\bar u_{1i}u_{3i}+\bar u_{1i}u_{3i}|^2\leq\frac14$.
\end{proof}

Note that by strengthening the latter inequalities to $J_{ii}=0$ we get the exact description of $\conv(\{J(P):P\in S_4\})$.
While the polytope bound from Lemma~\ref{lemma:spin-spin-upper-bound} is always larger than the conjectured bound, numerical results indicate that for different values of $\lambda$ the difference in volume does not exceed 10\%.

\section{MAJORIZATION THEOREM} \label{sec:majorization}

In qubit case the reduced state space $[0,1]\cong\Delta^1$ is one-dimensional, and hence there is only one optimal derivative for cooling (resp.\ heating).
In higher dimensions this is not true anymore. 
In this section we show that the set of achievable derivatives $\derv(\lambda)$ can be reduced to a subset of optimal derivatives, which typically still consists of more than one element.

An important way of comparing two mixed quantum states, or rather their eigenvalues, is called majorization~\cite{Horodecki13,Bengtsson17}.
First one defines majorization on vectors~\cite{MarshallOlkin}.
Let $\lambda,\mu\in\Delta^{n-1}$ be given.
Then $\lambda$ is said to \emph{majorize} $\mu$, denoted $\lambda\succeq\mu$ if 
$$
\textstyle\sum_{i=1}^k \lambda^\down_i \geq \sum_{i=1}^k \mu^\down_i, \quad k=1,\ldots,n,
$$
where $\lambda^\down_i$ is the $i$-th largest element of $\lambda$ and analogously for $\mu$.
The notion carries over to quantum states by defining that a state majorizes another if its eigenvalues majorize those of the other state. 
A function $f$ is \emph{Schur-convex} if $\lambda\succeq\mu$ implies $f(\lambda)\succeq f(\mu)$, and it is \emph{Schur-concave} if its negative is Schur-convex.
Due to fast unitary control we consider cost functions which depend only on the eigenvalues of the state.
Indeed, all the functions given in Section~\ref{sec:cost} are Schur-convex or Schur-concave functions of the eigenvalues.
This follows from the fact that these functions are convex or concave and invariant under permutations.
The following result, which specializes~\cite[Thm.~5.3]{Reduced23}, shows that for the purpose of cooling, a state which majorizes another is always better.

\begin{theorem} \label{thm:majorization}
Let $\mu:[0,\infty)\to\Delta^{n-1}$ be a solution to the relaxed control system~\eqref{eq:relaxed} and let $\lambda_0\in\Delta^{n-1}$ such that $\mu(0)=\mu_0\preceq \lambda_0$. 
Then there exists a solution $\lambda:[0,\infty)\to\Delta^{n-1}$ to~\eqref{eq:relaxed} with $\lambda(0)=\lambda_0$ such that $\mu(t)\preceq \lambda(t)$ for all $t\in[0,\infty)$.
\end{theorem}

Just like states are compared via majorization, two derivatives, at the same state, can be compared using an infinitesimal version of majorization, also called unordered majorization, cf.~\cite[Ex.~14.E.6]{MarshallOlkin}.
Let $v,w\in\R_0$ be two tangent vectors (derivatives) at $\lambda\in\Delta^{n-1}$.
Then we say that $v$ \emph{infinitesimally majorizes} $w$, denoted $v\tge w$ if
$$
\textstyle\sum_{i=1}^k v_i \geq \sum_{i=1}^k w_i, \quad k=1,\ldots,n.
$$
Theorem~\ref{thm:majorization} shows that derivatives which are not majorized by any other derivatives are \emph{optimal} for cooling.
If $\derv(\lambda)$ is a polytope, the subset of optimal derivatives takes a nice form.

\begin{corollary} \label{coro:opt-dervs}
Let $\lambda\in\Delta^{n-1}_\down$ be regular and assume that $\derv(\lambda)$ is a convex polytope.
If least one point in the relative interior of a face is optimal, then the entire face is optimal.
Moreover, the set of optimal faces is connected.
\end{corollary}

\begin{proof}
Let $P$ be the polytope $\derv(\lambda)$ and let $C$ be the cone of vectors infinitesimally majorizing the origin.
Then the optimal elements of $P$ are exactly the bounded faces of $P-C$,
and hence they form a subcomplex (cf.~\cite[Sec.~5.1]{Ziegler07}), which, by~\cite[Lem.~2.1]{Joswig01}, is connected.
\end{proof}

Furthermore, as we will see in the following section, for certain quasi-classical systems the optimal derivatives are always given by the same controls in the reduced control system, and do not depend on the state $\lambda$.

\section{OPTIMAL COOLING} \label{sec:opt-cooling}

Now we have all the tools we need to determine optimal cooling procedures for the $\mathrm V$-system and the spin-spin system (assuming Conjecture~\ref{conj:spin-spin}).
Note that in both cases Theorem~\ref{thm:asymptotic-coolability} shows immediately that the system is coolable.

\subsection{$\mathrm V$-system}

Using Corollary~\ref{coro:opt-dervs} one can show that the optimal derivatives are exactly the convex combinations of 
$$
\begin{pmatrix}
\gamma_2b+\gamma_1c\\
-\gamma_2b\\
-\gamma_1c
\end{pmatrix}
\text{ and }
\begin{pmatrix}
\gamma_1b+\gamma_2c\\
-\gamma_1b\\
-\gamma_2c
\end{pmatrix},
$$
where $\lambda=(a,b,c)\in\Delta^{n-1}_\down$. 
These correspond to the two topmost vertices in Figure~\ref{fig:inv-lambda-derivs}.
Note that the analogous result holds for the $\Lambda$-system, and hence our method also recovers the results of~\cite{Sklarz04} using a completely different approach.
Note that if $\gamma_1=\gamma_2$ the problem becomes trivial, so we assume that $\gamma_1<\gamma_2$.
A direct computation shows that the generators for these optimal derivatives commute and hence they can be applied in any order.
Hence following the first derivative for time $t_1$ and the second for time $t_2$ the final state is simply
$$
\begin{pmatrix}
a_0+(1-e^{-(\gamma_2t_1+\gamma_1t_2)})b_0+(1-e^{-(\gamma_1t_1+\gamma_2t_2)})c_0\\
e^{-(\gamma_2t_1+\gamma_1t_2)}b_0\\
e^{-(\gamma_1t_1+\gamma_2t_2)}c_0
\end{pmatrix}.
$$
Note that following these two derivatives it might happen that the second and third eigenvalue cross, but this does not change anything about the optimality of the derivatives.

In order to go any further we have to clarify the control task, since it necessarily takes infinite time to reach a pure state.
One natural choice is to minimize the time necessary to reach a certain largest eigenvalue, although other Schur-convex (or concave) functions such as those given in are also sensible.
Concretely the problem becomes, for any $0<\varepsilon<b_0+c_0$, to minimize $T=t_1+t_2$ under the conditions that $t_1,t_2\geq0$ and $e^{-(\gamma_2t_1+\gamma_1t_2)}b_0 + e^{-(\gamma_1t_1+\gamma_2t_2)}c_0=\varepsilon$.
Without the constraint $t_1,t_2\geq0$, an elementary computation shows that the optimal solution is
$$
\medmath{
t_1=\frac{\gamma_2\log(\frac{2b_0}\varepsilon)-\gamma_1\log(\frac{2c_0}\varepsilon)}{\gamma_2^2-\gamma_1^2}, \quad
t_2=\frac{\gamma_1\log(\frac{2b_0}\varepsilon)-\gamma_2\log(\frac{2c_0}\varepsilon)}{\gamma_1^2-\gamma_2^2}.}
$$
and the final state satisfies that $b(T)=c(T)=\varepsilon/2$.
However, $t_2$ becomes negative if $\varepsilon>2b_0(\frac{c_0}{b_0})^{\gamma_2/(\gamma_2-\gamma_1)}$.
In this case the optimal physical solution has $t_2=0$ and $t_1$ can be computed correspondingly.

This allows us to find the time-optimal controls for the task of reaching a largest eigenvalue of $1-\varepsilon$. 
We start by applying a (near) instantaneous unitary transformation to bring the state into diagonal form and with eigenvalues in weakly decreasing order. Then we wait for time $t_1$ without applying any controls (recall that the system is diagonally invariant and thus the compensating Hamiltonian vanishes). If $t_2=0$ we are done, otherwise we swap the second and third eigenvalue (near) instantaneously and wait for time $t_2$.

This solution is quite similar to that of~\cite{Sklarz04} for the \mbox{$\Lambda$-system}, except that we only switch the eigenvalues once.
Note also that in contrast to the approach of~\cite{Sklarz04}, we deduced the optimal solution instead of guessing it and we were able to prove optimality without the application of the Hamilton--Jacobi--Bellman equation.

\subsection{Spin-spin system}

Finally we determine an optimal cooling procedure for the spin-spin system.
Assuming Conjecture~\ref{conj:spin-spin} implies that the system is quasi-classical.
Corollary~\ref{coro:opt-dervs} allows us to determine the optimal derivatives.
It is easy to show that the only optimal derivatives are (convex combinations of)
$$
(b,-b,d,-d)^\top, \text{ and } (c,d,-c,-d)^\top.
$$

Conveniently, the corresponding generators again commute and hence their order of application is irrelevant.
Moreover it is clear that the system is diagonally invariant and hence the compensating Hamiltonian vanishes again.
Hence, similarly to the previous case, one can define a Schur-convex (or concave) cost function, such as purity, and find optimal times $t_1$ and $t_2$ via direct computation.
This is not much more difficult than in the previous case, but the resulting formulas are lengthy and not particularly enlightening, and hence omitted.

\section{CONCLUSION}

In this paper we applied the method of reduced control systems to Markovian quantum systems subject to fast unitary control in order to address the task of optimal cooling by determining the
corresponding provably time-optimal solutions.
This method allows for instance to derive a simple and efficiently computable characterization of asymptotically coolable systems.
Moreover using the Majorization Theorem, one can significantly simplify the search for optimal controls.
Concretely we studied three low dimensional systems, namely (i) general rank-one qubit systems, (ii) the $\mathrm V$-system on a qutrit and (iii) the spin-spin system consisting of two coupled qubits.
In each case we explicitly derived the optimal cooling solution and corresponding control functions.
The results go well beyond those obtained for instance in~\cite{Lapert10} and~\cite{Sklarz04}, 
all the while avoiding the Pontyagin Maximum Principle and the Hamilton--Jacobi--Bellman Equation.
At the same time this paper only scratches the surface of the topic of optimal cooling, and formulates a number of open problems, in particular a conjecture on the spin-spin system.
Although an essentially complete treatment of the qubit case is given in~\cite{QubitReduced24} by the author, the qutrit case, even restricted to systems with spontaneous emission, has so far only been solved in specific instances.
Overall this work presented novel 
powerful methods and demonstrated their effectiveness in concrete examples, paving the way for future advances also in experimental optimal cooling schemes.

\section{ACKNOWLEDGMENTS}

I would like to thank Thomas Schulte-Herbrüggen, Frederik vom Ende and Gunther Dirr as well as Brennan de Neeve and Florentin Reiter for their valuable feedback.

The project was funded i.a. by the Excellence Network of Bavaria under ExQM, by {\it Munich Quantum Valley} of the Bavarian State Government with funds from Hightech Agenda {\it Bayern Plus}.


\end{document}